\documentclass[superscriptaddress,aps,pra,onecolumn,showpacs,nofootinbib,longbibliography,notitlepage]{revtex4-1}
\usepackage{etex}
\usepackage{amsmath,amssymb,amsthm}
\usepackage[colorlinks=true,citecolor=blue,urlcolor=blue]{hyperref}
\usepackage[pdftex]{graphicx}
\usepackage{times,txfonts}
\usepackage{braket}
\usepackage{color}
\usepackage{natbib}
\usepackage{amsmath,blkarray}
\usepackage{mathtools}
\usepackage{latexsym}
\usepackage{tabularx, booktabs}
\usepackage{graphics,epstopdf}
\usepackage{graphicx}
\usepackage{float}
\usepackage{graphicx}
\usepackage{amsfonts}
\usepackage{subcaption}
\usepackage{color,soul}

\newcommand{\be}{\begin{equation}}
\newcommand{\ee}{\end{equation}}
\newcommand{\ba}{\begin{eqnarray}}
\newcommand{\ea}{\end{eqnarray}}

\newtheorem{theorem}{Theorem}

\begin{document}
\title{Role of maximally entangled states in the context of linear steering inequalities}

\author{Debarshi Das}
\email{dasdebarshi90@gmail.com}
\affiliation{Centre for Astroparticle Physics and Space Science (CAPSS), Bose Institute, Block EN, Sector V, Salt Lake, Kolkata 700 091, India}

\author{Souradeep Sasmal}
\email{souradeep.007@gmail.com}
\affiliation{Centre for Astroparticle Physics and Space Science (CAPSS), Bose Institute, Block EN, Sector V, Salt Lake, Kolkata 700 091, India}

\author{Arup Roy}
\email{arup145.roy@gmail.com}
\affiliation{S. N. Bose National Centre for Basic Sciences, Salt Lake, Kolkata 700 106, India}

\begin{abstract}
Linear steering inequalities are useful to check whether a bipartite state is steerable when both the parties are allowed to perform $n$ dichotomic measurements on their parts. In the present study we propose the necessary and sufficient condition under which $2$-settings linear steering inequality will be violated for any given set of spin-$\frac{1}{2}$ observables at trusted and untrusted parties' sides. The important result revealed by the present paper is that maximally entangled two-qubit states give the largest quantum violations of $2$-settings as well as $3$-settings linear steering inequalities attainable for any given set of spin-$\frac{1}{2}$ observables at trusted and untrusted parties' sides (if any violation exists for that given set of spin-$\frac{1}{2}$ observables).
\end{abstract}

\maketitle

\section{ Introduction}

In 1935 Einstein, Podolsky and Rosen (EPR) presented an argument showing the incompleteness of quantum mechanics \cite{EPR1935}. However, Schrodinger did not believe in that incompleteness. Rather EPR argument surprised him by the fact that an observer can control/steer a system which is not in her possession. This motivated Schrodinger to conceive the celebrated concept of `steering' \cite{S1935}. The concept of steering in the form of a task has been introduced recently \cite{WJD2007,JWD2007}. The task of quantum steering is to prepare different ensembles at one part of a bipartite system by performing local quantum measurements on another part in such a way that these ensembles cannot be explained by a local hidden state (LHS) model and no-signalling condition (the probability of obtaining one party's outcome does not depend on spatially separated other party's setting) is always satisfied by the bipartite system. This implies that the steerable correlations cannot be reproduced by a local hidden variable-local hidden state (LHV-LHS) model. In recent years, investigations related to quantum steering have been acquiring considerable significance, as evidenced by a wide range of studies \cite{st8,st10,steer3,st4,st9,sw,st5,sr,sf,se,new,sc,sg}. 

It is well-known that EPR steering lies in between entanglement and Bell nonlocality: Bell-nonlocal states form a strict subset of EPR steerable states which also form a strict subset of entangled states \cite{WJD2007,st11}. However, unlike Bell-nonlocality \cite{bell2} and entanglement \cite{ent}, the task of quantum steering is inherently asymmetric with respect to the observers \cite{st7}. In this case, the outcome statistics of one subsystem (which is being `steered') is due to valid quantum measurements on a valid quantum state. On the other hand, there is no such constraint for the other subsystem. There exist entangled states which are one-way steerable, i.e., demonstrate steerability from one observer to the other, but not vice-versa \cite{st7,ows1}. The study of quantum steering also finds applications in semi device independent scenario where the party, which is being ‘steered’, has trust on his/her quantum device but the other party's device is untrusted. One big advantage in this direction is that such scenarios are experimentally less demanding than fully device-independent protocols (where both of the parties distrust their devices) and, at the same time, require less assumptions than standard quantum cryptographic scenarios. Secure quantum key distribution (QKD) using quantum steering has been demonstrated \cite{st12}, where one party cannot trust his/her devices. 

	In \cite{CJWR2009} the authors have developed a series of ``linear steering inequalities" which are useful to check whether a bipartite state is steerable when both the parties are allowed to perform $n$ dichotomic measurements on his or her part. Apart from that several steering inequalities have been proposed \cite{si1,si2,si3,si4,si6,si7,si5,CFFW2015,GC2016,tighter,entrop1,entrop2} whose violations can render a correlation to be steerable. 
	
	In case of Bell-nonlocality, for an arbitrary given two qubit state, the maximum magnitude of the left hand side of Bell-CHSH (Bell-Clauser-Horne-Shimony-Holt) inequality \cite{bell,chsh} contingent upon using projective measurements of spin-$\frac{1}{2}$ observables has been studied \cite{hc}. On the other hand, in the context of EPR steering, the maximum magnitude of the left hand side of $2$-settings linear steering inequality \cite{CJWR2009} as well as that of EPR-steering analog of the CHSH inequality \cite{CFFW2015} for any given two qubit state under projective measurements of spin-$\frac{1}{2}$ observables have also been investigated \cite{costa,mal}. Furthermore, it has been shown that a given two qubit state violates $2$-settings linear steering inequality \textit{if and only if} the given state violates Bell-CHSH inequality and this is also true for EPR-steering analog of the CHSH inequality \cite{GC2016,costa}. In all these studies, maximum magnitudes of the left hand sides of $2$-settings linear steering inequality and Bell-CHSH inequality for a given two qubit state have been analysed by performing the maximization over all possible measurement settings. Motivated by the above results we investigate the maximum magnitude (maximized over all possible bipartite quantum states) of the left hand side of $2$-settings linear steering inequality attainable  for any given set spin-$\frac{1}{2}$ observables in the present study. Using this we propose the necessary and sufficient condition under which $2$-settings linear steering inequality will be violated for any given set of spin-$\frac{1}{2}$ observables at trusted and untrusted parties' sides. The maximum magnitude of the left hand side of Bell-CHSH inequality attainable for any given set of  spin-$\frac{1}{2}$ observables was studied by Kar \textit{et al.} \cite{guru}. By comparing these results we show that a given set of spin-$\frac{1}{2}$ observables violates $2$-settings linear steering inequality \textit{if and only if} that given set of spin-$\frac{1}{2}$ observables violates Bell-CHSH inequality.
	
	It was argued that, for any given set of spin-$\frac{1}{2}$ observables at the two spatially separated parties' sides, the maximum attainable quantum violation of Bell-CHSH inequality (if there exists any quantum violation for that given set of observables) is achieved if the shared state is a pure maximally entangled state \cite{guru,c}. Since, EPR steering has a vast application in semi device independent scenario as already discussed, it is important to study which entangled state is the most effective resource for witnessing EPR steering contingent upon using a specific set of observables. Here lies the motivation of the second part of our study. There are several inequalities to witness EPR steering \cite{CJWR2009,si1,si2,si3,si4,si6,si7,si5,CFFW2015,GC2016,tighter,entrop1,entrop2}. However, in the present study we restrict ourselves to the linear steering inequality \cite{CJWR2009} as this inequality can be used to probe EPR steering with \textit{arbitrary} number of dichotomic measurements on both sides. In particular, we address the following question:  which quantum states achieve the largest quantum violations of the $2$-settings and $3$-settings linear steering inequalities attainable for a given set of spin-$\frac{1}{2}$ observables (if there exists any quantum violation for that given set of observables).
	
	The plan of the paper is as follows. In Section \ref{secii} the basic
	notions of EPR steering and linear steering inequalities have been presented for the purpose of the present study. In Section \ref{seciii}, we
	present the necessary and sufficient condition under which $2$-settings linear steering inequality will be violated for any given set of spin-$\frac{1}{2}$ observables at trusted and untrusted parties' sides. In Section \ref{seciv}, we illustrate which quantum states provide the maximum quantum violations of the $2$-settings and $3$-settings linear steering inequalities attainable for any given set of spin-$\frac{1}{2}$ observables (if there exists any quantum violation for that given set of observables). Finally, in the concluding Section \ref{secv}, we elaborate a bit on the significance of the results obtained.

\section{ EPR Steering and Linear steering inequalities}\label{secii}
Let us recapitulate the concept of EPR steering as introduced by Wiseman et. al. \cite{WJD2007,JWD2007}. Let us consider that the joint state $\rho_{AB}$ of a pair of systems is shared between two spatially separated parties, say, Alice and Bob. Let $\mathcal{D_{\alpha}}$ and $\mathcal{D_{\beta}}$ denote the sets of observables in the Hilbert space of Alice's and Bob's systems, respectively. An element of $\mathcal{D_{\alpha}}$ is denoted by $A$, with a set of outcomes labeled by $a \in \mathcal{L}(A)$, and similarly an element of $\mathcal{D_{\beta}}$ is denoted by $B$, with a set of outcomes labeled by $b \in \mathcal{L}(B)$. The joint probability of obtaining the outcomes $a$ and $b$, when measurements $A$ and $B$ are performed locally by Alice and Bob on the joint state $\rho_{AB}$, respectively, is given by $P(a, b|A, B; \rho_{AB})$. The joint state $\rho_{AB}$ of the shared system is steerable by Alice to Bob \textit{iff} it is not the case that for all $a \in \mathcal{L}(A)$, $b \in \mathcal{L}(B)$, $A \in \mathcal{D_{\alpha}}$, $B \in \mathcal{D_{\beta}}$, the joint probability distributions can be written in the form,
\begin{equation}
\label{LHVLHS}
P(a, b|A, B; \rho_{AB}) = \sum_{\lambda} p(\lambda) P(a|A, \lambda) P(b|B, \rho_{\lambda}),
\end{equation}
where $p(\lambda)$ is the probability distribution over the hidden variables $\lambda$, $\sum_{\lambda} p(\lambda) = 1$; $P(a|A, \lambda)$ denotes an arbitrary probability distribution and $P(b|B, \rho_{\lambda})$ denotes the quantum probability of obtaining the outcome $b$ when measurement $B$ is performed on the quantum state (local hidden state) $\rho_{\lambda}$. In other words, the joint state $\rho_{AB}$ of the shared system will be called steerable if there is at least one measurement strategy for which the joint probability distribution does not satisfy a local hidden variable-local hidden state (LHV-LHS) model (\ref{LHVLHS}). One important point to be stressed here is that if for a given measurement
strategy the joint probability distribution has a LHV-LHS model, this does not imply that the joint state of the shared system is not steerable, since there could be another strategy that does not.

In \cite{CJWR2009} authors have constructed the following series of steering inequalities
to check whether a bipartite state is steerable from Alice to Bob when both the parties are allowed to perform $n$ dichotomic measurements on his or her part:
\begin{equation}
\label{linear_n}
F_n = \frac{1}{\sqrt{n}} \Bigl\lvert \sum_{i=1}^{n} \langle A_i \otimes B_i \rangle \Bigl\lvert \leq 1.
\end{equation}
These are called $n$-settings linear steering inequalities. The linear steering inequalities with $n=2$ and $n=3$ (which are relevant for spin-$\frac{1}{2}$ observables) are of the form:
\begin{equation}
\label{linear_2}
F_2 = \frac{1}{\sqrt{2}} \Bigl\lvert \sum_{i=1}^{2} \langle A_i \otimes B_i \rangle \Bigl\lvert \leq 1,
\end{equation}
and
\begin{equation}
\label{linear_3}
F_3 = \frac{1}{\sqrt{3}} \Bigl\lvert \sum_{i=1}^{3} \langle A_i \otimes B_i \rangle \Bigl\lvert \leq 1.
\end{equation}
Here $A_i = \hat{u}_i \cdot \vec{\sigma}$, $B_i = \hat{v}_i \cdot \vec{\sigma}$, $\vec{\sigma}$ = $(\sigma_1, \sigma_2, \sigma_3)$ is a vector composed of the Pauli matrices, $\hat{u}_i \in R^3$ are unit vectors, $\hat{v}_i \in R^3$ are orthonormal vectors, $\mu_2$ = $\{\hat{u}_1, \hat{u}_2, \hat{v}_1, \hat{v}_2\}$ is the set of measurement directions corresponding to $2$-settings linear steering inequality (\ref{linear_2}), $\mu_3$ = $\{\hat{u}_1, \hat{u}_2, \hat{u}_3, \hat{v}_1, \hat{v}_2, \hat{v}_3 \}$ is the set of measurement directions corresponding to $3$-settings linear steering inequality (\ref{linear_3}), $\langle A_i \otimes B_i \rangle$ = Tr$(\rho A_i \otimes B_i)$, where  $\rho \in \mathcal{H}_A \otimes \mathcal{H}_B$ is some bipartite quantum state shared between two spatially separated parties (Alice and Bob). Quantum violation of any of the above inequalities implies that the shared state is steerable from Alice to Bob. Note that the trusted party's (Bob) measurement directions are mutually orthogonal in case of $2$-settings and $3$-settings linear steering inequalities (\ref{linear_2} - \ref{linear_3}).

\section{ Spin-$\frac{1}{2}$ observables and $2$-settings linear steering inequality}\label{seciii}

In this Section we are going to investigate the necessary and sufficient condition for violating 2-settings linear steering inequality for any given set of spin-$\frac{1}{2}$ observables. We start by proposing the following theorem.

\begin{theorem} 
	For any given set of two spin-$\frac{1}{2}$ observables at Alice's side (untrusted party) and any given set of two spin-$\frac{1}{2}$ observables in mutually orthogonal directions at Bob's side (trusted party), there exists at least one  bipartite state for which the 2-settings linear steering inequality is violated \textit{iff} Alice's (untrusted party) spin-$\frac{1}{2}$ observables are non-commuting.
	\label{thm:0}
\end{theorem}

\begin{proof}
	
	The operator corresponding to $2$-settings linear steering inequality (\ref{linear_2}) can be written as,
	\begin{equation}
	\label{steope}
	O_{F_2} = \frac{1}{\sqrt{2}} (A_1 \otimes B_1 + A_2 \otimes B_2),
	\end{equation}
	where $A_i = \hat{u}_i \cdot \vec{\sigma}$, $B_i = \hat{v}_i \cdot \vec{\sigma}$ as discussed earlier. Hence, the square of the above steering operator is given by,
	\begin{equation}
	O^2_{F_2} = \frac{1}{2}(2 \mathbb{I} \otimes \mathbb{I} + A_1 A_2 \otimes B_1 B_2 + A_2 A_1 \otimes B_2 B_1 ).
	\end{equation}
	The above equation can be written in terms of the commutators and anti-commutators of the observables,
	\begin{equation}
	\label{steersq}
	O^2_{F_2} = \mathbb{I} \otimes \mathbb{I} + \frac{1}{4} (\{A_1, A_2\} \otimes \{B_1, B_2\}+[A_1, A_2] \otimes [B_1, B_2]),
	\end{equation}
	where $\{A_1, A_2\} = A_1 A_2 + A_2 A_1$ and $[A_1, A_2] = A_1 A_2 - A_2 A_1 $. Other commutators and anti-commutators are similarly defined. 
	
	For any two unit vectors $\hat{a}$ and $\hat{b}$, the commutation relation $[\hat{a} \cdot \vec{\sigma}, \hat{b} \cdot \vec{\sigma}]$ is given by,
	\begin{equation}
	[\hat{a} \cdot \vec{\sigma}, \hat{b} \cdot \vec{\sigma}] = 2 i (\hat{n}_{ab} \cdot \vec{\sigma} ) \sin \theta_{ab},
	\end{equation}
	where $\hat{n}_{ab}$ = $\dfrac{\hat{a} \times \hat{b}}{|\hat{a} \times \hat{b}|}$, is an unit vector perpendicular to the plane containing $\hat{a}$ and $\hat{b}$; $\theta_{ab}$ is the angle between the unit vectors $\hat{a}$ and $\hat{b}$. On the other hand, the anti-commutation relation $\{ \hat{a} \cdot \vec{\sigma}, \hat{b} \cdot \vec{\sigma} \}$ is given by,
	\begin{equation}
	\{ \hat{a} \cdot \vec{\sigma}, \hat{b} \cdot \vec{\sigma} \} = 2 \cos \theta_{ab} \mathbb{I}.
	\end{equation} 
	
	Using the above relations, Eq. (\ref{steersq}) can be rewritten as,
	\begin{align}
	O^2_{F_2}  =& \mathbb{I}\otimes \mathbb{I} + \cos\theta_{u_1 u_2} \cos\theta_{v_1 v_2} \mathbb{I} \otimes  \mathbb{I}  - \sin\theta_{u_1 u_2} \sin\theta_{v_1 v_2} (\hat{n}_{u_1 u_2} \cdot \vec{\sigma} ) \otimes  (\hat{n}_{v_1 v_2} \cdot \vec{\sigma} ),
	\end{align} 
	where $\theta_{u_1 u_2}$, $\theta_{v_1 v_2}$, $\hat{n}_{u_1 u_2}$, $\hat{n}_{v_1 v_2}$ are defined similarly. 
	
	Since the two spin-$\frac{1}{2}$ observables at trusted party's side are defined to be in orthogonal direction \cite{CJWR2009,costa}, i.e., $\theta_{v_1 v_2} = \frac{\pi}{2}$, we have,
	\begin{equation}
	\label{sqf}
	O^2_{F_2}  =\mathbb{I} \otimes \mathbb{I} - \sin\theta_{u_1 u_2} (\hat{n}_{u_1 u_2} \cdot \vec{\sigma} ) \otimes  (\hat{n}_{v_1 v_2} \cdot \vec{\sigma} ).
	\end{equation}
	Hence, the largest eigenvalue ($\lambda$) of $O^2_{F_2}$ is,
	\begin{equation}
	\label{ev}
	\lambda = 1 + |\sin\theta_{u_1 u_2}| 
	\end{equation}
	Corresponding to the eigenvalue (\ref{ev}), the largest eigenvalue ($\mu$) of the steering operator ($O_{F_2}$) is,
	\begin{equation}
	\label{lev}
	\mu = \sqrt{1 + |\sin\theta_{u_1 u_2}|}
	\end{equation}
	From Eq. (\ref{lev}) it is evident that for any given set of two spin-$\frac{1}{2}$ observables at Alice's side (untrusted party) and any given set of two spin-$\frac{1}{2}$ observables in mutually orthogonal directions at Bob's side (trusted party), there exists at least one bipartite qubit state for which the 2-settings linear steering inequality is violated \textit{iff} the largest eigenvalue $\mu$ of the steering operator $O_{F_2}$ is greater than $1$, i. e., $|\sin\theta_{u_1 u_2}| > 0$, i. e., $\sin\theta_{u_1 u_2} \neq 0$. Hence, we can conclude that for any given set of two spin-$\frac{1}{2}$ observables at Alice's side (untrusted party) and any given set of two spin-$\frac{1}{2}$ observables in mutually orthogonal directions at Bob's side (trusted party), there exists at least one bipartite qubit state for which the 2-settings linear steering inequality is violated \textit{iff} Alice's (untrusted party) two spin-$\frac{1}{2}$ observables are non-commuting.
\end{proof}

\textit{Corollary 1}: \textit{The largest quantum mechanical violation of $2$-settings inequality is given by $\sqrt{2}$} as the maximum value of the largest eigenvalue $\mu$ of the steering operator ($O_{F_2}$) is $\sqrt{2}$ when Alice's (untrusted party) spin-$\frac{1}{2}$ observables are in mutually orthogonal directions.\\

\textit{Corollary 2}: In Ref. \cite{guru}, the maximum magnitude of the left hand side of CHSH inequality was derived. For any given set of two spin-$\frac{1}{2}$ observables ($A_1$, $A_2$, $B_1$, $B_2$; $A_i = \hat{u}_i \cdot \vec{\sigma}$, $B_i = \hat{v}_i \cdot \vec{\sigma}$) at the spatially separated two party's (Alice and Bob) side, CHSH inequality will be violated \textit{iff} $|\sin\theta_{u_1 u_2} \sin\theta_{v_1 v_2}| > 0$ \cite{guru}. Hence, one can conclude that \textit{for any given set of two spin-$\frac{1}{2}$ observables at Alice's side (untrusted party) and any given set of two spin-$\frac{1}{2}$ observables in mutually orthogonal directions ($\theta_{v_1 v_2} = \frac{\pi}{2}$) at Bob's side (trusted party), the $2$-settings linear steering inequality will be violated iff CHSH inequality is violated with that given set of spin-$\frac{1}{2}$ observables at two spatially separated party's side.}

\section{Spin-$\frac{1}{2}$ observables and states which give maximum violations of $n$-settings linear steering inequalities}\label{seciv}

In this Section we are going to address the following question: which quantum states produce the maximum quantum violations of $2$-settings linear steering inequality (\ref{linear_2}) and $3$-settings linear steering inequality (\ref{linear_3}) attainable for any given set of spin-$\frac{1}{2}$ observables. We start by proposing the following theorem:

\begin{theorem}
	If there exists any quantum violation of the $2$-settings linear steering inequality (\ref{linear_2}) for any two given spin-$\frac{1}{2}$ observables at Alice's side and any two given spin-$\frac{1}{2}$ observables in mutually orthogonal directions at Bob's side, then that quantum violation reaches the maximum value (attainable for that given set of observables) when the shared state is maximally entangled two-qubit state.
	\label{th2}
\end{theorem}
\begin{proof}
	For any set of given spin-$\frac{1}{2}$ observables at Alice's and Bob's sides, the maximum attainable magnitude of the left hand side of $2$-settings linear steering inequality given by (\ref{linear_2}) must be achieved by some pure state as this state is the eigenstate corresponding to the largest eigenvalue of the operator (\ref{steope}) associated with $2$-settings linear steering for the given set of spin-$\frac{1}{2}$ observables. Furthermore, this eigenstate must be a two-qubit state as the operator (\ref{steope}) associated with $2$-settings linear steering for the given set of spin-$\frac{1}{2}$ observables belongs to the Hilbert space in $\mathbb{C}^2 \otimes \mathbb{C}^2$.  Any pure two-qubit state can be written in the following form, called the Schmidt decomposition \cite{sd1,sd2}:
	\begin{equation}
	\label{purestate}
	| \psi \rangle = \cos \alpha |00 \rangle + \sin \alpha |11 \rangle,
	\end{equation}
	where $0 \leq \alpha \leq \frac{\pi}{2}$, $\{ |0\rangle, |1\rangle \}$ is an orthonormal basis in the Hilbert space in $\mathcal{C}^2$. Let the operators corresponding to the given two spin-$\frac{1}{2}$ observables at Alice's side are $\hat{a}_1 \cdot \vec{\sigma}$ and $\hat{a}_2 \cdot \vec{\sigma}$ and that at Bob's side are $\hat{b}_1 \cdot \vec{\sigma}$ and $\hat{b}_2 \cdot \vec{\sigma}$. Any spin-$\frac{1}{2}$ observable can always be written as a linear combination of any three spin-$\frac{1}{2}$ observables in mutually orthogonal directions in the following way:
	\begin{align}
	\hat{a}_i \cdot \vec{\sigma} = & \sin \theta_i^a \cos \phi_i^a (\hat{m}_1 \cdot \vec{\sigma}) + \sin \theta_i^a \sin \phi_i^a (\hat{m}_2 \cdot \vec{\sigma}) + \cos \theta_i^a (\hat{m}_3 \cdot \vec{\sigma}),
	\label{ao}
	\end{align}
	and 
	\begin{align}
	\hat{b}_j \cdot \vec{\sigma} = & \sin \theta_j^b \cos \phi_j^b (\hat{m}_1 \cdot \vec{\sigma}) + \sin \theta_j^b \sin \phi_j^b (\hat{m}_2 \cdot \vec{\sigma}) + \cos \theta_j^b (\hat{m}_3 \cdot \vec{\sigma}),
	\label{bo}
	\end{align}
	where $(\hat{m}_1 \cdot \vec{\sigma})$, $(\hat{m}_2 \cdot \vec{\sigma})$ and $(\hat{m}_3 \cdot \vec{\sigma})$ are three spin-$\frac{1}{2}$ observables in mutually orthogonal directions. $0 \leq \theta_i^a \leq \pi$, $0 \leq \phi_i^a \leq 2 \pi$, $0 \leq \theta_j^b \leq \pi$, $0 \leq \phi_j^b \leq 2 \pi$. Let us construct the above three spin-$\frac{1}{2}$ observables in mutually orthogonal directions in the following way,
	\begin{equation}
	\label{ob1}
	(\hat{m}_1 \cdot \vec{\sigma}) = |0 \rangle \langle 0 | - |1 \rangle \langle 1|,
	\end{equation}
	\begin{equation}
	\label{ob2}
	(\hat{m}_2 \cdot \vec{\sigma}) = |+ \rangle \langle + | - |- \rangle \langle -|,
	\end{equation}
	and 
	\begin{equation}
	\label{ob3}
	(\hat{m}_3 \cdot \vec{\sigma}) = |\uparrow \rangle \langle \uparrow  | - |\downarrow \rangle \langle \downarrow |.
	\end{equation}
	Here, $\{ |+\rangle, |-\rangle \}$ is an orthonormal basis in the Hilbert space in $\mathcal{C}^2$ given by,
	\begin{equation}
	|+\rangle = \frac{1}{\sqrt{2}} ( |0\rangle + |1 \rangle),
	\end{equation}
	and 
	\begin{equation}
	|-\rangle = \frac{1}{\sqrt{2}} ( |0\rangle - |1 \rangle).
	\end{equation}
	$\{ |\uparrow\rangle, |\downarrow\rangle \}$ is another orthonormal basis in the Hilbert space in $\mathcal{C}^2$ given by,
	\begin{equation}
	|\uparrow\rangle = \frac{1}{\sqrt{2}} ( |0\rangle + i |1 \rangle),
	\end{equation}
	and 
	\begin{equation}
	|\downarrow\rangle = \frac{1}{\sqrt{2}} ( |0\rangle - i |1 \rangle).
	\end{equation}
	Now, for any two spin-$\frac{1}{2}$ observables $\hat{u} \cdot \vec{\sigma}$ and $\hat{v} \cdot \vec{\sigma}$, the anticommutation relation is given by,
	\begin{equation}
	\big\{\hat{u} \cdot \vec{\sigma} , \hat{v} \cdot \vec{\sigma} \big\} = 2 \hat{u} \cdot \hat{v} \mathbb{I}.
	\end{equation}
	From above Equation it is clear that $\big\{\hat{u} \cdot \vec{\sigma} , \hat{v} \cdot \vec{\sigma} \big\} = 0$ iff $\hat{u}$ and $\hat{v}$ are mutually orthogonal, i. e., $\hat{u} \cdot \vec{\sigma}$ and $\hat{v} \cdot \vec{\sigma}$ are two spin-$\frac{1}{2}$ observables in mutually orthogonal direction. Now, from Eqs.(\ref{ob1}), (\ref{ob2}), (\ref{ob3}) it can easily be checked that $\big\{\hat{m}_1 \cdot \vec{\sigma} , \hat{m}_2 \cdot \vec{\sigma} \big\}$ = $\big\{\hat{m}_1 \cdot \vec{\sigma} , \hat{m}_3 \cdot \vec{\sigma} \big\}$ = $\big\{\hat{m}_2 \cdot \vec{\sigma} , \hat{m}_3 \cdot \vec{\sigma} \big\}$ = $0$. Hence, $\hat{m}_1 \cdot \vec{\sigma}$, $\hat{m}_2 \cdot \vec{\sigma}$ and $\hat{m}_3 \cdot \vec{\sigma}$ given by Eqs.(\ref{ob1}), (\ref{ob2}) and (\ref{ob3}), respectively, are indeed three spin-$\frac{1}{2}$ observables in mutually orthogonal directions.

		Now consider that Alice and Bob share an arbitrary pure two qubit state (\ref{purestate}) as only a pure two-qubit state will achieve the maximum quantum violation of $2$-settings linear steering inequality (\ref{linear_2}) attainable for any given set of spin-$\frac{1}{2}$ observables. With these the left hand side of the $2$-settings linear steering inequality given by (\ref{linear_2}) becomes,
		\begin{align}
		F_2 = & \dfrac{1}{\sqrt{2}} \Big| \big[\cos \theta_1^a \cos \theta_1^b + \cos \theta_2^a \cos \theta_2^b \big]  + \big[ \cos ( \phi_1^a + \phi_1^b) \sin \theta_1^a \sin \theta_1^b  + \cos ( \phi_2^a + \phi_2^b) \sin \theta_2^a \sin \theta_2^b  \big]  \sin (2 \alpha) \Big|.
		\label{2s}
		\end{align}
		Note that $\sin (2 \alpha) \geq 0$ as $0 \leq \alpha \leq \frac{\pi}{2}$. There are the following possible cases, one of which will appear for any given set of observables:
		
		\textbf{Case I:} $\theta_i^a$, $\theta_j^b$, $\phi_i^a$, $\phi_j^b$ ($i, j = 1, 2$) are such that $\big[ \cos \theta_1^a \cos \theta_1^b + \cos \theta_2^a \cos \theta_2^b \big]$  $\geq 0$ and $ \big[\cos ( \phi_1^a + \phi_1^b) \sin \theta_1^a \sin \theta_1^b  + \cos ( \phi_2^a + \phi_2^b) \sin \theta_2^a \sin \theta_2^b \big] \geq 0$. In this case, from Eq.(\ref{2s}) it is clear that for any fixed values of $\theta_i^a$, $\theta_j^b$, $\phi_i^a$, $\phi_j^b$ ($i, j = 1, 2$), the left hand side of the $2$-settings linear steering inequality given by (\ref{linear_2}) will be maximized if $\alpha = \frac{\pi}{4}$. 
		
		\textbf{Case II:} $\theta_i^a$, $\theta_j^b$, $\phi_i^a$, $\phi_j^b$ ($i, j = 1, 2$) are such that $\big[ \cos \theta_1^a \cos \theta_1^b + \cos \theta_2^a \cos \theta_2^b \big]$  $\leq 0$ and $ \big[\cos ( \phi_1^a + \phi_1^b) \sin \theta_1^a \sin \theta_1^b  + \cos ( \phi_2^a + \phi_2^b) \sin \theta_2^a \sin \theta_2^b \big] \leq 0$. In this case also, for any fixed values of $\theta_i^a$, $\theta_j^b$, $\phi_i^a$, $\phi_j^b$ ($i, j = 1, 2$), the left hand side of the $2$-settings linear steering inequality given by (\ref{linear_2}) will be maximized if $\alpha = \frac{\pi}{4}$.
		
		\textbf{Case III-A:} $\theta_i^a$, $\theta_j^b$, $\phi_i^a$, $\phi_j^b$ ($i, j = 1, 2$) are such that $\big[ \cos \theta_1^a \cos \theta_1^b + \cos \theta_2^a \cos \theta_2^b \big]$  $\geq 0$; $ \big[\cos ( \phi_1^a + \phi_1^b) \sin \theta_1^a \sin \theta_1^b  + \cos ( \phi_2^a + \phi_2^b) \sin \theta_2^a \sin \theta_2^b \big] \leq 0$ and $\big| \cos \theta_1^a \cos \theta_1^b + \cos \theta_2^a \cos \theta_2^b \big|$  $\geq$ $\big|\cos ( \phi_1^a + \phi_1^b) \sin \theta_1^a \sin \theta_1^b  + \cos ( \phi_2^a + \phi_2^b) \sin \theta_2^a \sin \theta_2^b \big|$. In this case, from Eq.(\ref{2s}) we get $F_2$ = $\dfrac{1}{\sqrt{2}} \Big[ \big|\cos \theta_1^a \cos \theta_1^b + \cos \theta_2^a \cos \theta_2^b \big|  - \big| \cos ( \phi_1^a + \phi_1^b) \sin \theta_1^a \sin \theta_1^b  + \cos ( \phi_2^a + \phi_2^b) \sin \theta_2^a \sin \theta_2^b  \big|  \sin (2 \alpha) \Big]$. Hence, here for any fixed values of $\theta_i^a$, $\theta_j^b$, $\phi_i^a$, $\phi_j^b$ ($i, j = 1, 2$), the left hand side of the $2$-settings linear steering inequality given by (\ref{linear_2}) will be maximized if $\alpha = 0$. But $\alpha = 0$ implies that the shared state is separable. Therefore, this maximum magnitude of left  hand side of the $2$-settings linear steering inequality (\ref{linear_2}) does not lead to any quantum violation (i.e., this maximum magnitude is less than or equal to $1$) as any separable state is unsteerable and it cannot violate the $2$-settings linear steering inequality (\ref{linear_2}).
		
		\textbf{Case III-B:} $\theta_i^a$, $\theta_j^b$, $\phi_i^a$, $\phi_j^b$ ($i, j = 1, 2$) are such that $\big[ \cos \theta_1^a \cos \theta_1^b + \cos \theta_2^a \cos \theta_2^b \big]  \geq 0$; $ \big[\cos ( \phi_1^a + \phi_1^b) \sin \theta_1^a \sin \theta_1^b  + \cos ( \phi_2^a + \phi_2^b) \sin \theta_2^a \sin \theta_2^b \big] \leq 0$ and $\big| \cos \theta_1^a \cos \theta_1^b + \cos \theta_2^a \cos \theta_2^b \big|$  $\leq$ $\big|\cos ( \phi_1^a + \phi_1^b) \sin \theta_1^a \sin \theta_1^b  + \cos ( \phi_2^a + \phi_2^b) \sin \theta_2^a \sin \theta_2^b \big|$. In this case, from Eq.(\ref{2s}) we get $F_2$ = $\dfrac{1}{\sqrt{2}} \Big[ - \big|\cos \theta_1^a \cos \theta_1^b + \cos \theta_2^a \cos \theta_2^b \big|  + \big| \cos ( \phi_1^a + \phi_1^b) \sin \theta_1^a \sin \theta_1^b  + \cos ( \phi_2^a + \phi_2^b) \sin \theta_2^a \sin \theta_2^b  \big|  \sin (2 \alpha) \Big]$. Hence, here for any fixed values of $\theta_i^a$, $\theta_j^b$, $\phi_i^a$, $\phi_j^b$ ($i, j = 1, 2$), the left hand side of the $2$-settings linear steering inequality given by (\ref{linear_2}) will be maximized if $\alpha =\frac{\pi}{4}$.
		
		\textbf{Case IV-A:} $\theta_i^a$, $\theta_j^b$, $\phi_i^a$, $\phi_j^b$ ($i, j = 1, 2$) are such that $\big[ \cos \theta_1^a \cos \theta_1^b + \cos \theta_2^a \cos \theta_2^b \big]  \leq 0$; $ \big[\cos ( \phi_1^a + \phi_1^b) \sin \theta_1^a \sin \theta_1^b  + \cos ( \phi_2^a + \phi_2^b) \sin \theta_2^a \sin \theta_2^b \big] \geq 0$ and $\big| \cos \theta_1^a \cos \theta_1^b + \cos \theta_2^a \cos \theta_2^b \big|$  $\geq$ $\big|\cos ( \phi_1^a + \phi_1^b) \sin \theta_1^a \sin \theta_1^b  + \cos ( \phi_2^a + \phi_2^b) \sin \theta_2^a \sin \theta_2^b \big|$. Following the argument presented in Case III-A, we can state that, for any fixed values of $\theta_i^a$, $\theta_j^b$, $\phi_i^a$, $\phi_j^b$ ($i, j = 1, 2$), the left hand side of the $2$-settings linear steering inequality given by (\ref{linear_2}) will be maximized if $\alpha = 0$, which does not correspond to any quantum violation of the $2$-settings linear steering inequality (\ref{linear_2}).
		
		\textbf{Case IV-B:} $\theta_i^a$, $\theta_j^b$, $\phi_i^a$, $\phi_j^b$ ($i, j = 1, 2$) are such that $\big[ \cos \theta_1^a \cos \theta_1^b + \cos \theta_2^a \cos \theta_2^b \big]  \leq 0$; $ \big[\cos ( \phi_1^a + \phi_1^b) \sin \theta_1^a \sin \theta_1^b  + \cos ( \phi_2^a + \phi_2^b) \sin \theta_2^a \sin \theta_2^b \big] \geq 0$ and $\big| \cos \theta_1^a \cos \theta_1^b + \cos \theta_2^a \cos \theta_2^b \big|$  $\leq$ $\big|\cos ( \phi_1^a + \phi_1^b) \sin \theta_1^a \sin \theta_1^b  + \cos ( \phi_2^a + \phi_2^b) \sin \theta_2^a \sin \theta_2^b \big|$. In this case we follow the argument presented in Case III-B. Here for any fixed values of $\theta_i^a$, $\theta_j^b$, $\phi_i^a$, $\phi_j^b$ ($i, j = 1, 2$), the left hand side of the $2$-settings linear steering inequality given by (\ref{linear_2}) will be maximized if $\alpha =\frac{\pi}{4}$.
		
		Note that in Eq.(\ref{2s}) we have not assumed that the two spin-$\frac{1}{2}$ observables at Bob's side are in mutually orthogonal directions. Hence, the above result holds even if we assume that $\hat{b}_1 \cdot \hat{b}_2$ = $0$.
	
\end{proof}

Now we are going to address the aforementioned question in the context of $3$-settings linear steering inequality (\ref{linear_3}).  In this case, we propose the following theorem:
\begin{theorem}
	If there exists any quantum violation of the $3$-settings linear steering inequality (\ref{linear_3}) for any three given spin-$\frac{1}{2}$ observables at Alice's side and any three given spin-$\frac{1}{2}$ observables in mutually orthogonal directions at Bob's side, then that quantum violation reaches the maximum value (attainable for that given set of observables) when the shared state is maximally entangled two-qubit state.
	\label{th3}
\end{theorem}
\begin{proof}
	Following the similar argument presented in the proof of Theorem \ref{th2} we can state that the largest magnitude of the left hand side of $3$-settings linear steering inequality (\ref{linear_3})  attainable for any given set of spin-$\frac{1}{2}$ observables on the spatially separated two parties' sides must be achieved by some pure two-qubit state. Let us consider that Alice and Bob share an arbitrary pure two qubit state (\ref{purestate}). The operators corresponding to the given three spin-$\frac{1}{2}$ observables at Alice's side are $\hat{a}_1 \cdot \vec{\sigma}$, $\hat{a}_2 \cdot \vec{\sigma}$ and $\hat{a}_3 \cdot \vec{\sigma}$, where $\hat{a}_i \cdot \vec{\sigma}$ ($i = 1, 2, 3$) is given by Eq.(\ref{ao}). On the other hand, the operators corresponding to the given three spin-$\frac{1}{2}$ observables at Bob's side are $\hat{b}_1 \cdot \vec{\sigma}$, $\hat{b}_2 \cdot \vec{\sigma}$ and $\hat{b}_3 \cdot \vec{\sigma}$, where $\hat{b}_j \cdot \vec{\sigma}$ ($j = 1, 2, 3$) is given by Eq.(\ref{bo}). With these the left hand side of the $3$-settings linear steering inequality given by (\ref{linear_3}) becomes,
		\begin{align}
		F_3 = & \dfrac{1}{\sqrt{3}} \Big| \big[ \cos \theta_1^a \cos \theta_1^b + \cos \theta_2^a \cos \theta_2^b + \cos \theta_3^a \cos \theta_3^b \big]  + \big[ \cos ( \phi_1^a + \phi_1^b) \sin \theta_1^a \sin \theta_1^b  + \cos ( \phi_2^a + \phi_2^b) \sin \theta_2^a \sin \theta_2^b \nonumber \\
		& +  \cos ( \phi_3^a + \phi_3^b) \sin \theta_3^a \sin \theta_3^b  \big] \sin (2 \alpha) \Big|. 
		\label{3s}
		\end{align}
		In this case also the following cases appear:
		
		\textbf{Case I:} $\theta_i^a$, $\theta_j^b$, $\phi_i^a$, $\phi_j^b$ ($i, j = 1, 2, 3$) are such that $\big[ \cos \theta_1^a \cos \theta_1^b$ $+ \cos \theta_2^a$ $\cos \theta_2^b$ $+ \cos \theta_3^a \cos \theta_3^b \big]$  $\geq 0$ and $\big[ \cos ( \phi_1^a + \phi_1^b) \sin \theta_1^a \sin \theta_1^b$  $+ \cos ( \phi_2^a$ $+ \phi_2^b)$ $\sin \theta_2^a$ $\sin \theta_2^b +  \cos ( \phi_3^a + \phi_3^b) \sin \theta_3^a \sin \theta_3^b  \big]$ $\geq 0$. 
		
		\textbf{Case II:} $\theta_i^a$, $\theta_j^b$, $\phi_i^a$, $\phi_j^b$ ($i, j = 1, 2, 3$) are such that $\big[ \cos \theta_1^a \cos \theta_1^b + \cos \theta_2^a$ $\cos \theta_2^b + \cos \theta_3^a \cos \theta_3^b \big]  \leq 0$ and $\big[ \cos ( \phi_1^a + \phi_1^b) \sin \theta_1^a \sin \theta_1^b  + \cos ( \phi_2^a + \phi_2^b) \sin \theta_2^a$ $\sin \theta_2^b +  \cos ( \phi_3^a + \phi_3^b) \sin \theta_3^a \sin \theta_3^b  \big] \leq 0$. 
		
		\textbf{Case III-A:} $\theta_i^a$, $\theta_j^b$, $\phi_i^a$, $\phi_j^b$ ($i, j = 1, 2,3$) are such that $\big[ \cos \theta_1^a \cos \theta_1^b + \cos \theta_2^a \cos \theta_2^b + \cos \theta_3^a \cos \theta_3^b \big]  \geq 0$, $\big[ \cos ( \phi_1^a + \phi_1^b) \sin \theta_1^a \sin \theta_1^b  + \cos ( \phi_2^a + \phi_2^b) \sin \theta_2^a$ $\sin \theta_2^b +  \cos ( \phi_3^a + \phi_3^b) \sin \theta_3^a \sin \theta_3^b  \big] \leq 0$ and $\big| \cos \theta_1^a \cos \theta_1^b + \cos \theta_2^a \cos \theta_2^b + \cos \theta_3^a$ $\cos \theta_3^b \big|$  $\geq$ $\big| \cos ( \phi_1^a + \phi_1^b) \sin \theta_1^a \sin \theta_1^b  + \cos ( \phi_2^a + \phi_2^b) \sin \theta_2^a \sin \theta_2^b +  \cos ( \phi_3^a + \phi_3^b)$ $\sin \theta_3^a$ $\sin \theta_3^b  \big|$. 
		
		\textbf{Case III-B:} $\theta_i^a$, $\theta_j^b$, $\phi_i^a$, $\phi_j^b$ ($i, j = 1, 2,3$) are such that $\big[ \cos \theta_1^a \cos \theta_1^b + \cos \theta_2^a \cos \theta_2^b + \cos \theta_3^a \cos \theta_3^b \big]  \geq 0$, $\big[ \cos ( \phi_1^a + \phi_1^b) \sin \theta_1^a \sin \theta_1^b  + \cos ( \phi_2^a + \phi_2^b) \sin \theta_2^a$ $\sin \theta_2^b +  \cos ( \phi_3^a + \phi_3^b) \sin \theta_3^a \sin \theta_3^b  \big] \leq 0$ and $\big| \cos \theta_1^a \cos \theta_1^b + \cos \theta_2^a \cos \theta_2^b + \cos \theta_3^a$ $\cos \theta_3^b \big|$  $\leq$ $\big| \cos ( \phi_1^a + \phi_1^b) \sin \theta_1^a \sin \theta_1^b  + \cos ( \phi_2^a + \phi_2^b) \sin \theta_2^a \sin \theta_2^b +  \cos ( \phi_3^a + \phi_3^b)$ $\sin \theta_3^a$ $\sin \theta_3^b  \big|$. 
		
		\textbf{Case IV-A:} $\theta_i^a$, $\theta_j^b$, $\phi_i^a$, $\phi_j^b$ ($i, j = 1, 2,3$) are such that $\big[ \cos \theta_1^a \cos \theta_1^b + \cos \theta_2^a \cos \theta_2^b + \cos \theta_3^a \cos \theta_3^b \big]  \leq 0$, $\big[ \cos ( \phi_1^a + \phi_1^b) \sin \theta_1^a \sin \theta_1^b  + \cos ( \phi_2^a + \phi_2^b)$ $\sin \theta_2^a \sin \theta_2^b +  \cos ( \phi_3^a + \phi_3^b) \sin \theta_3^a \sin \theta_3^b  \big] \geq 0$ and $\big| \cos \theta_1^a \cos \theta_1^b + \cos \theta_2^a \cos \theta_2^b + \cos \theta_3^a \cos \theta_3^b \big|$  $\geq$ $\big| \cos ( \phi_1^a + \phi_1^b) \sin \theta_1^a \sin \theta_1^b  + \cos ( \phi_2^a + \phi_2^b) \sin \theta_2^a \sin \theta_2^b +  \cos ( \phi_3^a + \phi_3^b) \sin \theta_3^a \sin \theta_3^b  \big|$. 
		
		\textbf{Case IV-B:} $\theta_i^a$, $\theta_j^b$, $\phi_i^a$, $\phi_j^b$ ($i, j = 1, 2,3$) are such that $\big[ \cos \theta_1^a \cos \theta_1^b + \cos \theta_2^a \cos \theta_2^b + \cos \theta_3^a \cos \theta_3^b \big]  \leq 0$, $\big[ \cos ( \phi_1^a + \phi_1^b) \sin \theta_1^a \sin \theta_1^b  + \cos ( \phi_2^a + \phi_2^b)$ $\sin \theta_2^a \sin \theta_2^b +  \cos ( \phi_3^a + \phi_3^b) \sin \theta_3^a \sin \theta_3^b  \big] \geq 0$ and $\big| \cos \theta_1^a \cos \theta_1^b + \cos \theta_2^a \cos \theta_2^b + \cos \theta_3^a \cos \theta_3^b \big|$  $\leq$ $\big| \cos ( \phi_1^a + \phi_1^b) \sin \theta_1^a \sin \theta_1^b  + \cos ( \phi_2^a + \phi_2^b) \sin \theta_2^a \sin \theta_2^b +  \cos ( \phi_3^a + \phi_3^b) \sin \theta_3^a \sin \theta_3^b  \big|$. 
		
		In all these case we follow the same argument presented in the proof of Theorem \ref{th2}. For any given set of spin-$\frac{1}{2}$ observables we obtain the following: either i) maximum quantum violation of $3$-settings linear steering inequality (\ref{linear_3}) is achieved when the shared two-qubit state is maximally entangled, or ii) there does not exist any violation of the $3$-settings linear steering inequality (\ref{linear_3}).
		
		Note that in Eq.(\ref{3s}) we have not assumed that the three spin-$\frac{1}{2}$ observables at Bob's side are in mutually orthogonal directions. Hence, the above result holds even if we assume that $\hat{b}_1 \cdot \hat{b}_2$ = $0$; $\hat{b}_2 \cdot \hat{b}_3$ = $0$; $\hat{b}_1 \cdot \hat{b}_3$ = $0$.

\end{proof}

\section{ Conclusion}\label{secv}
In the present study we have calculated the maximum magnitude (maximum over all possible two-qubit states) of the left hand side of $2$-settings linear steering inequality attainable for any given set of spin-$\frac{1}{2}$ observables on the spatially separated two parties' sides. It has also been shown that a given set of spin-$\frac{1}{2}$ observables violates $2$-settings linear inequality \textit{if and only if} the given set of spin-$\frac{1}{2}$ observables violates Bell-CHSH inequality. Note that it was earlier shown that any given two-qubit state violates $2$-settings linear steering inequality \textit{if and only of} that state violates Bell-CHSH inequality \cite{costa}. Hence, the result presented in this study complements the result obtained in the previous studies \cite{GC2016,costa}.

There are several inequalities which are useful for showing EPR steering \cite{CJWR2009,si1,si2,si3,si4,si6,si7,si5,CFFW2015,GC2016,tighter,entrop1,entrop2}. Since the entangled states that demonstrate EPR steering are proved to be useful resources for various semi-device independent quantum informational tasks, it is important to investigate which state maximally violates a steering inequality for any given set of observables on the spatially separated two parties' sides. In the present study, restricting ourselves to the  linear steering inequalities \cite{CJWR2009} and spin-$\frac{1}{2}$ observables, we have addressed the above issue. In particular, we have shown that the pure maximally entangled two-qubit states give the largest attainable quantum violations of $2$-settings and $3$-settings linear steering inequalities among all possible quantum states for any given set of spin-$\frac{1}{2}$ observables (if there exists any violation for that given set of spin-$\frac{1}{2}$ observables).

It is well known that Bell nonlocality and EPR steering are operationally inequivalent. However, for a given two-qubit state these two notions are equivalent in $2 - 2 - 2$ experimental scenario (involving two parties, two measurement settings per party, two outcomes per setting) in the sense that a given two-qubit state demonstrates steering in this scenario \textit{if and only if} the given two-qubit state shows Bell nonlocality in this scenario  \cite{GC2016,costa}. Motivated by these facts, we have investigated in the present study whether there exists any non-equivalence between Bell nonlocality and EPR steering when the set of spin-$\frac{1}{2}$ observables on the spatially separated two parties' sides is fixed in the above scenario. However, the present results demonstrate the equivalence between Bell nonlocality and EPR steering in the above context.

Addressing the above questions in the context of other steering criterion is worth for future research. We have seen that maximally entangled states give largest violation of linear steering inequality attainable for any given set of spin-$\frac{1}{2}$ observables at trusted and untrusted parties' sides (if there exists any violation for that given set of spin-$\frac{1}{2}$ observables). Moreover, it can also be conjectured that any steering inequality whose left hand side is a linear function of correlations will show the same feature. However, there are lots of linear local realist inequalities which are optimally violated by some pure states other than maximally entangled states \cite{cglmp,tiltedchsh}. Hence, it will be interesting to find out steering inequalities which are maximally violated by some pure states other than maximally entangled states for any given set of spin-$\frac{1}{2}$ observables at trusted and untrusted parties' sides. Another direction for future research will be posing the same question for non-linear steering inequalities and entropic steering inequalities \cite{CFFW2015,GC2016,entrop1,entrop2} proposed so far. Moving beyond spin-$\frac{1}{2}$ observables and two-qubit states it is legitimate to ask which quantum state maximally violates (if there exists a violation) a steering inequality for any set of observables on both sides. Though the results presented in this paper have shown similarities between Bell nonlocality  and EPR steering in the context of state space structure, the above mentioned questions may demonstrate non-equivalence between Bell nonlocality and EPR steering.\\

\section{Acknowledgements}
We would like to gratefully acknowledge fruitful discussions with Prof. Guruprasad Kar and  Some Sankar Bhattacharya. DD acknowledges the financial support from University Grants Commission (UGC), Government of India. SS acknowledges the financial support from INSPIRE programme, Department of Science and Technology, Government of India.

\end{document}